\documentclass[runningheads,a4paper]{llncs}
\usepackage{subcaption}
\usepackage{xcolor}
\usepackage{hyperref}
\usepackage{tikz}
\usetikzlibrary{arrows,shapes}
\usetikzlibrary{calc}

\usepackage{amsmath}
\usepackage{amssymb}
\usepackage{graphicx}
\usepackage{enumitem}

\usepackage{stackrel}
\usepackage{mathtools}
\usepackage{enumitem}

\usepackage[dvips]{epsfig}
\usepackage[english]{babel}
\usepackage[T1]{fontenc}
\usepackage[utf8]{inputenc}
\usepackage{dsfont}

\usepackage[algo2e,algoruled,ruled,vlined,linesnumbered]{algorithm2e}

\usepackage{url}

\begin{document}

\newcommand{\refmap}{\curlyvee}
\newcommand{\corefmap}{\curlywedge}
\newcommand{\ELIMINE}[1]{}

\mainmatter  %
\title{Morse Sequences on Stacks and Flooding Sequences}

\titlerunning{Morse Sequences on Stacks and Flooding Sequences}

\author{Gilles Bertrand}
%
\authorrunning{G. Bertrand}
%
\institute{Univ Gustave Eiffel, CNRS, LIGM, F-77454 Marne-la-Vallée, France 
\email{gilles.bertrand@esiee.fr}
}

\index{Bertrand, Gilles}

\maketitle

\newcommand{\bbbu}{\; \ddot{\cup} \;}
\newcommand{\axr}[1]{\ddot{\textsc{#1}}  \normalsize}

\newcommand{\axcup}{\textsc{C\tiny{UP}} }
\newcommand{\axcap}{\textsc{C\tiny{AP}} }
\newcommand{\axunion}{\textsc{U\tiny{NION}} }
\newcommand{\axinter}{\textsc{I\tiny{NTER}} }

\newcommand{\bb}[1]{\mathbb{#1}}
\newcommand{\ca}[1]{\mathcal{#1}}
\newcommand{\ax}[1]{\textsc{#1} \normalsize}

\newcommand{\axb}[2]{\ddot{\textsc{#1}} \textsc{\tiny{#2}}  \normalsize}
\newcommand{\bbb}[1]{\ddot{\mathbb{#1}}}
\newcommand{\cab}[1]{\ddot{\mathcal{#1}}}

\newcommand{\rel}[1]{\scriptstyle{\mathbf{#1}}}

\newcommand{\rela}[1]{\textsc{\scriptsize{\bf{{#1}}}} \normalsize}

\newcommand{\de}[2]{#1[#2]}
\newcommand{\di}[2]{#1\langle #2 \rangle}

\newcommand{\la}{\langle}
\newcommand{\ra}{\rangle}
\newcommand{\hs}{\hspace*{\fill}}

\newcommand{\cell}{\mathbb{C}}
\newcommand{\cellp}{\mathbb{C}^\times}
\newcommand{\simp}{\mathbb{S}}
\newcommand{\comp}{\mathbb{H}}
\newcommand{\simpp}{\mathbb{S}^\times}
\newcommand{\den}{\mathbb{D}\mathrm{en}}
\newcommand{\ram}{\mathbb{R}\mathrm{am}}
\newcommand{\tree}{\mathbb{T}\mathrm{ree}}
\newcommand{\graph}{\mathbb{G}\mathrm{raph}}
\newcommand{\vertex}{\mathbb{V}\mathrm{ert}}
\newcommand{\edge}{\mathbb{E}\mathrm{dge}}
\newcommand{\equ}{\mathbb{E}\mathrm{qu}}
\newcommand{\esub}{\mathbb{E}\mathrm{sub}}

\newcommand{\topp}{\langle \mathrm{K} \rangle}
\newcommand{\topq}{\langle \mathrm{Q} \rangle}
\newcommand{\topxp}{\langle \mathbb{X}, \mathrm{P} \rangle}
\newcommand{\topxq}{\langle \mathbb{X},\mathrm{Q} \rangle}

\newcommand{\vt}{\mathcal{K}}
\newcommand{\vtp}{\mathcal{T'}}
\newcommand{\vtpp}{\mathcal{T''}}
\newcommand{\vq}{\mathcal{Q}}
\newcommand{\vqp}{\mathcal{Q'}}
\newcommand{\vqpp}{\mathcal{Q''}}
\newcommand{\vk}{\mathcal{K}}
\newcommand{\vkp}{\mathcal{K'}}
\newcommand{\vkpp}{\mathcal{K''}}

\newcommand{\C}{\ensuremath{\searrow^{\!\!\!\!\!C}}}
\newcommand{\Detach}{\ensuremath{\;\oslash\;}}

\newcommand{\mh}[1]{\hat{\textsl{#1}}}

\newcommand{\sig}{\sigma}
\newcommand{\del}{W}

\newcommand{\ms}{\overrightarrow{W}}
\newcommand{\msp}{\overrightarrow{W'}}
\newcommand{\mspp}{\overrightarrow{W_p}}

\newcommand{\msi}{\overrightarrow{W_{i}}}
\newcommand{\msim}{\overrightarrow{W_{i-1}}}
\newcommand{\mss}{\widehat{W}}
\newcommand{\msc}{\ddot{W}}
\newcommand{\mso}{\overleftarrow{W}}

\newcommand{\ov}{\overline}
\newcommand{\un}{\underline}
\newcommand{\ha}{\widehat}
\newcommand{\ti}{\widetilde}
\newcommand{\dd}{\ddot}

\newcommand{\ka}{\kappa}
\newcommand{\ova}{\overrightarrow{w}}


\begin{abstract}

This paper builds upon the framework of \emph{Morse sequences}, a simple and effective approach to discrete Morse theory. A Morse sequence on a simplicial complex consists of a sequence of nested subcomplexes generated by expansions and fillings—two operations originally introduced by Whitehead. Expansions preserve homotopy, while fillings introduce critical simplexes that capture essential topological features.

We extend the notion of Morse sequences to \emph{stacks}, which are monotonic functions defined on simplicial complexes,
and define \emph{Morse sequences on stacks}
 as those whose expansions preserve the homotopy of all sublevel sets. This extension leads to a generalization of the fundamental collapse 
 theorem to weighted simplicial complexes. Within this framework, we focus on a refined class of sequences called \emph{flooding sequences}, 
 which exhibit an ordering behavior similar to that of classical watershed algorithms. Although not every Morse sequence on a stack is a flooding sequence, 
 we show that the gradient vector field associated with any Morse sequence can be recovered through a flooding sequence.

Finally, we present algorithmic schemes for computing flooding sequences using cosimplicial complexes.

\end{abstract}

\section{Introduction}

This work is set within the framework of \emph{Morse sequences}~\cite{bertrand2025morse} offering a simple approach to discrete Morse theory~\cite{For98a}.
If $L$ is a subcompex of a simplicial complex $K$, a Morse sequence from $L$ to $K$  is a sequence 
$\ms = \langle L = K_0,\ldots,K_k =K \rangle$ of nested simplicial complexes
such that, for each $i \in [1,k]$,
$K_i$ is either an \emph{expansion} or a \emph{filling} of $K_{i-1}$. 
These two elementary  operations were introduced in \cite{Whi39}: 
expansions (the inverse of collapses) preserve homotopy, while fillings introduce critical simplices. 
The set of all expansions of the sequence induces a  gradient vector field between $L$ and $K$, while the set of all critical simplexes 
captures fundamental homological properties of the pair $(L,K)$.

In this paper, we extend these notions by exploring Morse sequences constructed on stacks defined on simplicial complexes (Section \ref{sec:Fsequence}).
A \emph{stack} on a simplicial complex $K$ is a monotonic function $F$ assigning a value $F(\nu)$ to each simplex $\nu \in K$. 
A \emph{Morse sequence on $F$} is a Morse sequence on $K$ in which the expansions preserve the homotopy of all sublevel sets  of $F$. 
This naturally generalizes the fundamental collapse theorem of discrete Morse theory to weighted simplicial complexes.

We introduce and study a refined class of Morse sequences on stacks, called \emph{flooding sequences} (Section \ref{sec:flood}). 
These sequences satisfy an ordering condition analogous to those found in some classical  watershed algorithms.
A Morse sequence on a stack is not necessarily  a flooding sequence. 
Nevertheless, we prove that the gradient vector field of an arbitrary Morse sequence on a stack $F$
can be obtained by computing a flooding sequence on $F$.

Thanks to the notion of a \emph{cosimplicial complex} (a complex which is the difference of two simplicial complexes), we derive simple schemes for computing flooding sequences (Section \ref{sec:floodcosim} and \ref{sec:compflood}). In particular, we consider
a \emph{maximal increasing scheme} (building $\ms$ from $L$ to $K$ by prioritizing expansions), and a \emph{minimal decreasing scheme} (building
$\ms$ from $K$ to $L$ by prioritizing collapses).

These methods follow a propagation paradigm aimed at minimizing the number of critical simplexes.
We show that they yield efficient algorithms for computing flooding sequences.

\section{Basic definitions}
\label{sec:cosimplicial}

\subsection{Simplicial and cosimplicial complexes}

In this section, we introduce \emph{cosimplicial complexes}, which generalize the classical notion of simplicial complexes. 
In the context of combinatorial dynamics, these structures are also referred to as \emph{multivectors} \cite{Dey22}.
In this paper, we employ  cosimplicial complexes to construct the sequences of simplicial complexes that form
Morse sequences.

By a {\it simplex}, we mean a finite non-empty set. 
The {\it dimension} of a simplex~$\nu$, written $dim(\nu)$, is the number of its elements
minus one.

\begin{definition} \label{def:cosim0}
Let $S$ be a finite set of simplexes. The set $S$ is a {\em cosimplicial complex} if,
for any $\sig,\tau \in S$, we have $\nu \in S$ whenever $\sig \subseteq \nu \subseteq \tau$. 
\end{definition} 
\noindent

A simplex of a cosimplicial complex $S$ is {\it a face of~$S$}.
A face of $S$ is a {\em facet of~$S$} if it is maximal for inclusion, and a {\em cofacet of~$S$} if it is minimal for inclusion.

A finite set $K$  
of simplexes is a {\it simplicial complex} if,
for each $\tau \in K$, we have $\sig \in K$ whenever $\sig \subseteq \tau$ and $\sig \not= \emptyset$. 
Thus, each simplicial complex is also a cosimplicial complex. \\

Let $K$ be a simplicial complex and let $S \subseteq K$. We say that $S$ is \emph{closed for $K$} if $S$ is a simplicial complex.
We say that $S$ is \emph{open for $K$} if, for each $\sig \in S$, we have $\tau \in S$ whenever $\sig \subseteq \tau$ and $\tau \in K$. 
Therefore, a set $S \subseteq K$ is closed for $K$ if and only if $K \setminus S$ is open for $K$.
The family $\mathcal{O}_K$ composed of all sets which are open for $K$ satisfies the axioms of a topology on $K$. 
So we obtain the structure of a finite topological space, see \cite{Alex37}, \cite{Bar11}. We easily get the following.

\begin{proposition} \label{prop:cloopen}
Let $K$ be a simplicial complex and let $S \subseteq K$. 
The set $S$ is a cosimplicial complex if and only if 
$S$ is the intersection of an open set for $K$ and a closed set for $K$. 
\end{proposition}

In general topology, a subset which satisfies the condition of Proposition \ref{prop:cloopen} is said to be  \emph{locally closed}, see \cite{bourbaki2013general},
Ch.1, §3, no3.

Let $R$ and $T$ be two subsets of $K$. We have $S = R \cap T$
if and only if $S = R \setminus R'$, with  $R' = K \setminus T$. Thus, by Proposition \ref{prop:cloopen} and by the properties of the sets which are open or closed for $K$,
we obtain the following.

 \begin{proposition} \label{prop:bourbaki}
Let $K$ be a simplicial complex and let $S \subseteq K$. 
\begin{itemize}[topsep=0cm] 
\item $S$ is a cosimplicial complex iff $S$ is the difference of two open sets for $K$.
\item $S$ is a cosimplicial complex iff $S$ is the difference of two closed sets for $K$.
\end{itemize}
\end{proposition}

 Since a closed set for $K$ is a simplicial complex in its own right, we derive:
 
  \begin{proposition} \label{prop:bourbaki1}
  A set $S$ is a cosimplicial complex if and only if $S$ is the difference of two simplicial complexes.
\end{proposition}
 
 \subsection{An ambient space}

In \cite{Knu24}, an \emph{open simplicial complex} is defined precisely as the difference of two simplicial complexes.
Consequently, by Proposition \ref{prop:bourbaki1}, cosimplicial complexes and open simplicial complexes are equivalent notions.
In this paper, we will take advantage of the local characterization of Definition \ref{def:cosim0}
for deriving some operators which act directly on a cosimplicial complex.

Let $S$ be a finite set of simplexes. 
Let $\overline{S}$ be the set of simplexes such that
$\sig \in \overline{S}$ if and only if there exists $\tau \in S$ with $\sig \subseteq \tau$. 
The set $\overline{S}$ is the \emph{closure of $S$}. 
In the sequel of this paper, we write $\underline{S} = \overline{S} \setminus S$.

We have $S \subseteq \overline{S}$ and 
we observe that
$\overline{S}$ is a simplicial complex, this complex is the smallest simplicial complex that contains $S$.
It follows that $S$ is a simplicial complex if and only if $\overline{S} = S$.
Also, the set $\underline{S}$ is the smallest set such that $S \cup \underline{S}$ is a simplicial complex.

By the previous remarks, the simplicial complex $\overline{S}$ may be considered as ``an ambient topological space for $S$''. 
As a direct consequence of the definition of a cosimplicial complex we obtain:

\begin{proposition} \label{prop:cosim0}
Let $S$ be a finite set of simplexes.\\
1) The set $S$ is a cosimplicial complex if and only if $S$ is open for $\overline{S}$. \\
2) The set $S$ is a cosimplicial complex if and only if $\underline{S}$ is a simplicial complex.
\end{proposition}

If $S$ is a cosimplicial complex, then the sets $L = \underline{S}$ and $K = \overline{S}$ are two simplicial complexes such that 
$L \subseteq K$ and $S = K \setminus L$. The following proposition generalizes this situation. 

\begin{proposition} \label{prop:cosim1} Let $(L,K)$ be a pair of simplicial complexes such that $L \subseteq K$.
The cosimplicial complex $S = K \setminus L$ is such that
 $L \cup \overline{S} = K$ and $L \cap \overline{S} = \underline{S}$. 
\end{proposition}

Let $S$ be a cosimplicial complex and let $\nu \in S$. We write: \\
-  $\partial (\nu,S) = \{ \eta \in S \; | \; \eta \subseteq \nu$ and $dim(\eta) =dim(\nu)-1 \}$, and \\
- $\delta (\nu,S) = \{ \mu \in S \; | \; \nu \subseteq \mu$ and $dim(\mu) =dim(\nu)+1 \}$; \\
$\partial (\nu,S)$ and $\delta (\nu,S)$ are, respectively, the {\it boundary} and the {\it coboundary of $\nu$ in $S$.} \\
We observe that: \\
- the face $\nu$ is a facet of $S$ if and only if $\delta (\nu,S) = \emptyset$, \\
- the face $\nu$ is a cofacet of $S$ if and only if $\partial (\nu,S) = \emptyset$.


\subsection{Stacks and filtrations}

With the notion of a stack, we introduce some weights on a cosimplicial or a simplicial complex. 
A stack is just a monotone function that assigns values to the simplices of the complex.

Let $F$ be a map from a cosimplicial complex $S$ to $\bb{Z}$. We say that $F$ is a {\em stack on $S$} if
we have $F(\sigma) \leq F(\tau)$ whenever $\sigma,\tau \in S$ and $\sigma\subseteq \tau$.

Let $F$ be a map from a cosimplicial complex $S$ to $\bb{Z}$.
 For any $\lambda \in \bb{Z}$, we write: \\
 \hspace*{\fill}
  $F_\lambda = \{\nu \in S \; \mid \; F(\nu) \leq \lambda \}$ and
$F[\lambda] = \{\nu \in S \; \mid \; F(\nu) = \lambda \}$, \hspace*{\fill} \\
$F_\lambda$ and $F[\lambda]$ are, respectively, the {\em cut} and {\em the section of $F$ at level $\lambda$}.\\

If $F$ is a stack on a cosimplicial complex $S$, then the indexed family $(F_\lambda)_{\lambda\in\bb{Z}}$ is a \emph{filtration on $S$}. 
That is we have $F_\lambda \subseteq F_{\lambda'}$ whenever $\lambda \leq \lambda'$.
Also if $F$ is a stack on a simplicial complex, then any cut of $F$ is a simplicial complex. 
Furthermore: 

\begin{proposition} \label{prop:stack2}
If $F$ is a stack on a cosimplicial complex, 
then any cut of $F$ and any section of $F$ is a cosimplicial complex.
\end{proposition}

\section{Morse sequences on stacks}
\label{sec:Fsequence}

We recall the definitions of simplicial collapses and simplicial expansions~\cite{Whi39}. \\
Let $K$ be a simplicial complex and let $\sig, \tau \in K$.
The couple $(\sig,\tau)$ is a {\em free pair for~$K$}, 
if $\tau$ is the only face of $K$ that contains $\sig$.
If $(\sig,\tau)$ is a free 
pair for $K$, then the simplicial complex $L = K \setminus \{ \sig,\tau \}$
is {\em an elementary 
collapse of $K$}, and $K$ is {\em an elementary 
expansion of $L$}. 
We say that
$K$ {\em collapses onto $L$},
or that $L$ {\em expands onto $K$},
if there exists a  sequence 
$\langle K=K_0,\ldots,K_k=L \rangle$, such that
$K_i$ is an elementary collapse of $K_{i-1}$, $i \in [1,k]$.

We  recall also the definitions of perforations and fillings \cite{Whi39}. \\
Let $K,L$ be simplicial complexes.
If $\nu \in K$ is a facet of $K$ and if $L = K \setminus \{\nu \}$, we say that
$L$ is {\em an elementary perforation of $K$}, and that
$K$ is {\em an elementary filling of $L$}.

We now introduce the notion of a ``Morse sequence'' by simply considering expansions and fillings
of a simplicial complex \cite{bertrand2025morse}.

\begin{definition} \label{def:seq1}
Let $L\subseteq K$ be two simplicial complexes. A \emph{Morse sequence (from $L$ to $K$)} is a sequence
$\ms = \langle L = K_0,\ldots,K_k =K \rangle$ of simplicial complexes
such that,
for each $i \in [1,k]$, $K_i$ is either an elementary expansion or an elementary filling of $K_{i-1}$.
If $L=\emptyset$, we say that $\ms$ is a \emph{Morse sequence on $K$}.
Thus, the sequence $\ms = \langle K \rangle$ is trivially a Morse sequence on $K$.
\end{definition}

Note that a Morse sequence $\ms = \langle K_0,...,K_k \rangle$ is a filtration, that is,
for each $i \in [0,k-1]$, we have $K_i \subseteq K_{i+1}$; see \cite{Sco19,Edels13}. 

If $\ms = \langle \emptyset = K_0,...,K_k = K \rangle$ is a Morse sequence on $K$, with $k \geq 1$, we also note that
$K_1$ is necessarily a filling of $\emptyset$. 

\noindent
Let $\ms = \langle L=K_0,\ldots,K_k=K \rangle$ be a Morse sequence from $L$ to $K$.
If $k \geq 1$, we write   
$\diamond \ms$ for the sequence 
$\diamond \ms = \langle \ka_1, \ldots, \ka_k \rangle$ such that, for each $i \in [1,k]$:
\begin{itemize}[noitemsep,topsep=0pt]
\item If $K_i$ is an elementary filling of $K_{i-1}$, then $\ka_i$ is the simplex such that $K_i = K_{i-1} \cup \{\ka_i\}$; we say that
$K_i$ and the face $\ka_i$ are \emph{critical for $\ms$}. 
\item If $K_i$ is an elementary expansion of $K_{i-1}$, then
$\ka_i$  is the free pair $(\sigma,\tau)$  such that $K_i = K_{i-1} \cup \{\sigma,\tau \}$; we say that $K_i$, the pair
$\ka_i$, and the faces $\sig$, $\tau$, are \emph{regular for $\ms$}. 
\end{itemize}
If $k=0$, we write $\diamond \ms = \langle \rangle$. That is, $\diamond \ms$ is the empty sequence. \\
We say that
$\diamond \ms$ is a \emph{simplex-wise (Morse) sequence (from $L$ to $K$)}.
Note that $\diamond \ms$ is a sequence of faces and pairs.

\begin{definition} \label{def:seq2}
Let $\ms$ be a Morse sequence. 
The \emph{gradient vector field of 
$\ms$} is the set composed of all regular pairs for $\ms$.
We say that two Morse sequences $\overrightarrow{V}$ and $\ms$
from $L$ to $K$
are \emph{equivalent} if
they have the same gradient vector field.
\end{definition}


Let us consider the simplicial complex $K$ depicted Fig. \ref{fig:BasicFlooding} (a). 
For convenience, we will describe a simplex by a concatenation of its vertices. \\
Let $\ms = \langle \emptyset = K_0,...,K_7 = K \rangle$ be the sequence such that 
$K_1 = \{a\}$, $K_2 = K _1 \cup \{b,ab\}$,  $K_3 = K _2 \cup \{c,bc\}$,  $K_4 = K _3 \cup \{d,cd\}$,
 $K_5 = K _4 \cup \{e,ed\}$,  $K_6 = K _5 \cup \{be\}$,  $K_7 = K _6 \cup \{bd,bde\}$.
It can be seen that $\ms$ is a Morse sequence on $K$. The corresponding simplex-wise sequence $\diamond \ms$ 
is such that $\diamond \ms = \langle a, (b, ab), (c,bc), (d,cd),(e,ed),be,(bd,bde) \rangle$.
In Fig. \ref{fig:BasicFlooding} (b)
the gradient vector field of $\ms$ is given by arrows, the two critical faces $a$ and $be$ of $\ms$ are in red.\\

 Now, we extend the notion of a Morse sequence for an arbitrary stack $F$. \\
\medskip
\hspace*{\fill}
{\it In the sequel of this paper, $L$ and $K$ will denote simplicial complexes.}
 \hspace*{\fill}

Let $F$ be a stack on $K$ and let $L$ be a subcomplex of $K$. Let $\ka = (\sigma,\tau)$ be a free pair for~$L$.
If $F(\sigma) = F(\tau)$, then we say that
$L' = L \setminus \{\sigma,\tau \}$ is an \emph{(elementary) $F$-collapse of $L$}
and $L$ is an \emph{(elementary) $F$-expansion of $L'$}. 
We write $F(\ka) = F(\sigma) = F(\tau)$.

\begin{figure}[tb]
    \centering
     \hfill
    \begin{subfigure}{0.25\textwidth}
        \includegraphics[width=\textwidth]{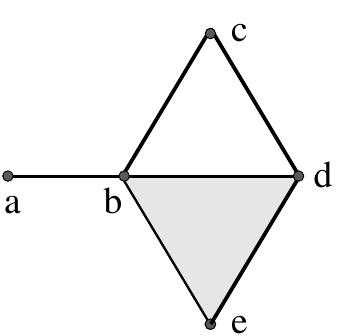}
        \caption{}
    \end{subfigure}%
    \hfill
    \begin{subfigure}{0.25\textwidth}
         \includegraphics[width=\textwidth]{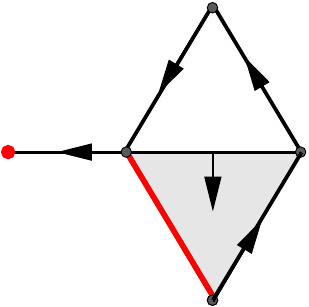}
        \caption{}
    \end{subfigure}%
     \hfill. \\
    
    \begin{subfigure}{0.25\textwidth}
         \includegraphics[width=\textwidth]{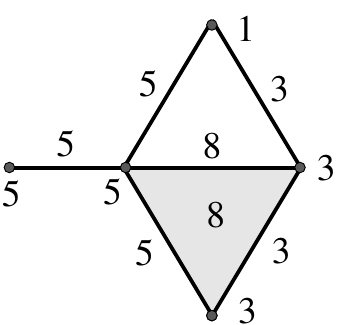}
        \caption{}
    \end{subfigure}%
    \hfill
      \begin{subfigure}{0.25\textwidth}
         \includegraphics[width=\textwidth]{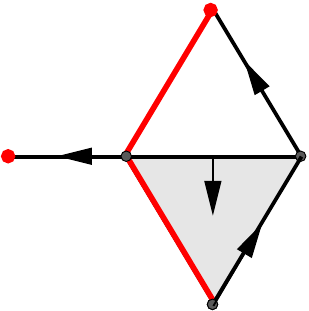}
        \caption{}
    \end{subfigure}%
     \hfill
       \begin{subfigure}{0.25\textwidth}
         \includegraphics[width=\textwidth]{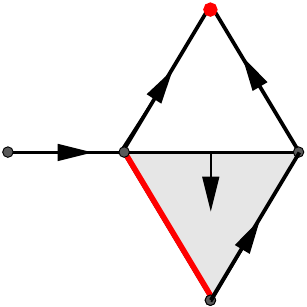}
        \caption{}
    \end{subfigure}%
     \hfill

 \caption{ (a) A simplicial complex $K$. (b) The gradient vector field of a Morse sequence $\ms$ on $K$. 
 (c) A stack $F$ on $K$. (d) The gradient vector field of a Morse sequence $\ms'$ on $F$. 
 (e) The gradient vector field of a flooding sequence $\ms''$ on~$F$.
 See text for the description of the sequences $\ms$, $\ms'$, and $\ms''$.}
  \label{fig:BasicFlooding}
\end{figure}

\begin{definition} \label{def:stack2}
Let $\ms$ be a Morse sequence from $L$ to $K$ and let $F$ be a stack on $K$. 
We say that $\ms$ is a \emph{(Morse) sequence on $F$} or {\em an $F$-sequence (from $L$ to $K$)}, if we have $F(\sig) = F(\tau)$ whenever 
$(\sig,\tau)$ is a regular pair for $\ms$.
\end{definition}

An example of a stack $F$ on a complex $K$ is given Fig. \ref{fig:BasicFlooding} (c). 
Let us consider again the Morse sequence $\ms$ and its 
gradient vector field given in $(b)$. We see that $\ms$ does not satisfy the condition of Def. \ref{def:stack2} for the stack $F$. Thus, the sequence $\ms$ 
is not an $F$-sequence. Now let us consider the sequence $\diamond \ms' = \langle a, (b, ab), c, bc, (d,cd),(e,ed), be,(bd,bde) \rangle$.
This sequence is a simplex-wise Morse sequence on $K$. Furthermore the corresponding sequence
$\ms'$ is an $F$-sequence on $K$. The gradient vector field of $\ms'$ is depicted in (d), the four critical faces $a$, $c$, $bc$, and $be$ 
of $\ms'$ are in red.
\\

Let $F$ be a stack on $K$ and $L$ be a subcomplex of $K$. 
 If $\lambda \in \mathbb{Z}$, the set $F_\lambda \cap L$ is the cut of $L$ at level $\lambda$ which is induced by $F$. 
 Note that $F_\lambda \cap K = F_\lambda$. \\
Let $(\sigma,\tau)$ be a free pair for $L$ such that $F(\sig) = F(\tau)$ . 
Then: \\
- For each $\lambda < F(\sigma)$, we have $\sigma \not\in F_\lambda \cap L$ and $\tau \not\in F_\lambda \cap L$, \\
- For each $\lambda \geq F(\sigma)$, one can verify that the pair $(\sigma,\tau)$ is a free pair for $F_\lambda \cap L$.

Now, let $\ms  = \langle L = K_0,..., K_k = K \rangle$ be an $F$-sequence. 
Suppose $K_i$ is an elementary expansion of $K_{i-1}$.
Then, by the above observation, this elementary expansion is a (possibly trivial) expansion for all cuts of $K_{i-1}$. \\
By induction, we have:

\begin{theorem} \label{prop:stack4} Let $F$ be a stack on $K$ and
 $\ms = \langle L = K_0,\ldots,K_k = K \rangle$ be an $F$-sequence.
Let $K_i$ and $K_j$, $j >i$, be two consecutive
critical complexes for  $\ms$. 
That is, $K_{i+1},...,K_{j-1}$ are regular complexes for $\ms$.
Then,  for each $\lambda \in \mathbb{Z}$, the complex $F_\lambda \cap K_{j-1}$ 
collapses onto $F_\lambda \cap K_{i}$.
\end{theorem}

This last property may be seen as an extension to stacks of a fundamental theorem, called \emph{the collapse theorem}, 
which makes the link between the basic definitions of  discrete Morse theory
and discrete homotopy (See Theorem 3.3 of \cite{For98a} and Theorem 4.27 of \cite{Sco19}).

\begin{remark}
Let $\ms = \langle \emptyset = K_0,..., K_k = K \rangle$
be an $F$-sequence. Then $\ms$ induces a ``double filtration'': 
the indexed family $(K_i)_{i \in [0,k]}$ is a  filtration where each $K_i$ induces the filtration $(F_\lambda \cap K_i)_{\lambda \in \mathbb{Z}}$.  
Also, $\ms$ induces the sequence $\langle F_0,..., F_k \rangle$ where each $F_i$ is the stack on $K_i$ which is the restriction
of $F$ to $K_i$.  
\end{remark}

If $S$ is a set, we write   $\mathds{1}_S$ for the constant function $S \rightarrow \mathbb{Z}$
such that, for each $x \in S$, we have $\mathds{1}_S(x) = 1$. 
We observe that a sequence $\ms$  is a Morse sequence from $L$ to $K$ if and only if $\ms$ is an $F$-sequence from $L$ to $K$ for $F = \mathds{1}_K$.\\

It is well known that discrete Morse theory can be formulated from the perspective of gradient vector fields~\cite{koz20}. In the remainder of this section, we present a natural extension of this classical notion to stacks. We then emphasize that a Morse sequence is in no way a restrictive notion with respect to this point of view.
That is, any gradient vector field on a stack $F$ is the gradient vector field of some Morse sequence on $F$.

Let $F$ be a stack on $K$ and let 
$V$ be a set of pairs $(\sigma,\tau)$, with $\sigma, \tau \in K$, such that
$\sigma \in \partial \tau$ and $F(\sigma) = F(\tau)$.
We say that $V$ is a  \emph{(discrete) vector field on $F$}
if each simplex of $K$ is in at most one pair of $V$.
A simplex $\nu \in K$ is said to be \emph{critical for $V$} if it does not belong to any pair in $V$.

Let $V$ be a vector field on $F$.
A \emph{($p$-)gradient path in $V$ (from $\sig_0$ to $\sig_k$)}
is a sequence
$\pi = \langle \sig_0,\tau_0,\sig_1,\tau_1,...,\sig_{k-1},\tau_{k-1}, \sigma_{k}\rangle$,
with $k \geq 0$, composed of faces 
such that, for all $i \in [0,k-1]$, $dim(\sig_i) =dim(\sig_k)= p$,
$(\sig_i,\tau_{i})$ is in~$V$,
$\sigma_{i+1} \subset \tau_i$, and $\sig_{i+1} \not= \sig_i$.
This sequence $\pi$ is said to be \emph{trivial} if $k =0$, that is, if 
$\pi = \langle \sig_0 \rangle$; otherwise, if $k \geq 1$, we say that $\pi$ is \emph{non-trivial}.
Also, the sequence $\pi$ is \emph{closed} if $\sig_0 = \sig_k$.
We say that a vector field $V$ on $F$ is \emph{a gradient vector field} if $V$ is \emph{acyclic}, that is if $V$ contains
no non-trivial closed $p$-gradient path.

Note that, if $F$ is the constant map $\mathds{1}_K$, then the above definitions correspond exactly to the 
classical ones.
The proof of the following result is a direct extension of the proof given in~\cite{bertrand2025morse} (Theorem 51).

\begin{theorem} \label{th:DVFS} Let $F$ be a stack on $K$.
A vector field $V$ on $F$ is a gradient vector field if and only if $V$ is the gradient vector field of a
Morse sequence on $F$.
\end{theorem}


 \section{Flooding sequences}
 \label{sec:flood}

 We introduce the following refinement of an $F$-sequence.
 A flooding sequence corresponds to a process similar to the one often used for computing the watershed transform, which is a key tool for image segmentation in Mathematical Morphology
 \cite{beucher1979}, \cite{BM93}.

 \begin{definition} \label{def:cut0}
Let $F$ be a stack on $K$ and let $\ms = \langle \emptyset = K_0,..., K_k = K \rangle$ be an $F$-sequence. 
We say that $\ms$ is {\em a flooding sequence (on $F$)}, if we have $F(\sigma) \leq F(\tau)$ whenever $\sigma \in K_i$, $\tau \in K_j \setminus K_i$, and $i < j$.
\end{definition}

\begin{figure*}[tb]
    \centering
    \begin{subfigure}[t]{0.45\textwidth}
        \centering
        \includegraphics[width=.99\textwidth]{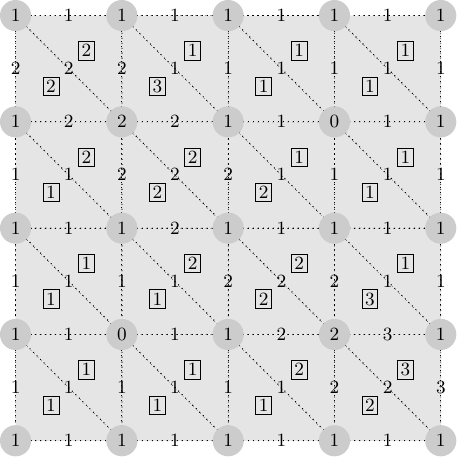}
        \caption{}
    \end{subfigure}%
    ~~~~~~~~~~
    \begin{subfigure}[t]{0.45\textwidth}
        \centering
         \includegraphics[width=.99\textwidth]{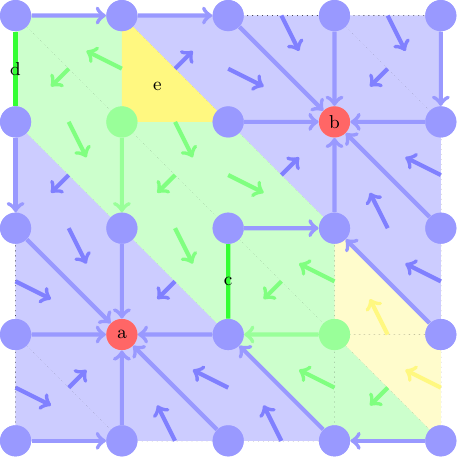}
        \caption{}
    \end{subfigure}%

 \caption{A flooding sequence.  (a) a simplicial stack $F$ on a triangulation $K$ of a square. (b) A flooding sequence on $F$. }
  \label{fig:FloodingSequence}
\end{figure*}

In  the example of Fig. \ref{fig:BasicFlooding}, it can be seen that the sequence $\ms'$ is not a flooding sequence on $F$. The first element of $\ms'$ is the 
face $a$. In a flooding sequence, the face $c$ should appear before $a$ since $F(a) = 5$ and $F(c)= 1$.
 Now, let us consider the sequence $\diamond \ms'' = \langle c, (d,cd),(e,ed), (b,bc), (a, ab), be,(bd,bde) \rangle$.
The corresponding sequence $\ms''$ is a flooding sequence on $F$. The gradient vector field of $\ms''$ is given in (e), the two critical faces $c$ and $be$ 
of $\ms''$ are in red.

Figure~\ref{fig:FloodingSequence} illustrates an example of a flooding sequence on a less simple complex. 
This sequence $\ms$ can begin from two possible points, both at level 0, as shown in Figure~\ref{fig:FloodingSequence}.a. 
If $\ms$ begins with \textbf{a}, then the second element of $\ms$ must be \textbf{b},
both of them are critical for $\ms$.
Then the sequence must continue with a subsequence composed of all simplexes at level 1. This subsequence is illustrated in blue, it turns out that it is composed 
solely of regular elements. The next subsequence composed of all simplexes at level 2 is illustrated in green, it contains two critical simplexes \textbf{c}
and \textbf{d}. The last subsequence composed of all simplexes at level 3 is illustrated in yellow, it contains one critical simplex \textbf{e}. 
See also Figure 1 in \cite{BN25} which gives, on the same complex, two examples of $F$-sequences that are not flooding sequences.

Let $F$ be a stack on $K$ and let $L$ be a subcomplex of $K$. We say that $L$
is \emph{a subcomplex for $F$} if, for all $\sigma \in K$ and all $\tau \in L$,
we have $\sigma \in L$ whenever $F(\sigma) < F(\tau)$. 
The following provides two characterizations of flooding sequences.

\begin{proposition} \label{prop:cut1}
Let $F$ be a stack on $K$ and let $\ms = \langle \emptyset = K_0,..., K_k = K \rangle$ be an $F$-sequence. \\
1) The sequence $\ms$ is a flooding sequence on $F$ if and only if, for each $i \in [0,k]$, 
the complex $K_i$ is a subcomplex for $F$. \\
2) The sequence $\ms$ is a flooding sequence on $F$ if and only if, for each $\lambda \in \bb{Z}$, there exists $i \in [0,k]$ such that $F_\lambda = K_i$.
\end{proposition}

The next lemma will allow us to transform $F$-sequences into flooding sequences
by means of an exchange property.

\begin{lemma} \label{prop:stack4}
Let $F$ be a stack on $K$, and let 
$\diamond \ms = \langle \ka_1, \ldots, \ka_i,\ka_{i+1},\ldots ,\ka_k \rangle$, 
with $k \geq 2$, be a simplex-wise $F$-sequence. We consider the sequence 
$\diamond \ms' = \langle \ka'_1, \ldots, \ka'_i,\ka'_{i+1},\ldots ,\ka'_k \rangle$, 
such that $\ka'_i = \ka_{i+1}$, $\ka'_{i+1} = \ka_i$, and $\ka'_j = \ka_{j}$ for all $j \not=i$, $j\not=i+1$.
If $F(\ka_i) > F(\ka_{i+1})$, with  $1 \leq i \leq k-1$, then $\diamond \ms'$ is a simplex-wise $F$-sequence.
\end{lemma}

\begin{proof} We write $\ms = \langle \emptyset = K_0,..., K_k = K \rangle$ for the Morse sequence 
corresponding to $\diamond \ms$, and $\ms' = \langle \emptyset = K'_0,..., K'_k = K \rangle$
for the one corresponding to $\diamond \ms'$.
We consider the different cases for $\ka_i$ and $\ka_{i+1}$:
\begin{enumerate} [topsep=0cm]
\item Suppose $\ka_i$ is a critical simplex; that is,  $K_i$ a filling of $K_{i-1}$.  
\begin{enumerate} [topsep=0cm]
 \item Suppose $\ka_{i+1}$ is a critical simplex; that is $K_{i+1}$ is a filling of $K_{i}$. \\
- If $\ka_{i} \not\in \partial(\ka_{i+1})$, then we see that $K'_i = K_{i-1} \cup\{\ka_{i+1}\}$ is a filling of $K_{i-1}$, and 
$K_{i+1} = K'_i \cup \{ \ka_{i} \}$ is a filling of $K'_i$. Thus $\diamond \ms'$ is a simplex-wise $F$-sequence. \\
- If $\ka_{i} \in \partial(\ka_{i+1})$, then $\ms'$ is not a filtration. But, in this case, 
we have  $F(\ka_{i}) \leq F(\ka_{i+1})$ since $F$ is a stack on $K$. 
\item Suppose $\ka_{i+1} = (\sigma, \tau)$; that is $K_{i+1}$ is an expansion of $K_{i}$. The pair
$(\sigma, \tau)$ is a free pair for $F$, we have $F(\ka_{i+1}) = F(\sigma) = F(\tau)$.\\
- If $\ka_{i} \not\in \partial(\tau)$, then 
$K'_i = K_{i-1} \cup\{\sigma,\tau\}$ is an expansion of $K_{i-1}$, and 
$K_{i+1} = K'_i \cup \{ \ka_{i} \}$ is a filling of $K'_i$. Thus $\diamond \ms'$ is a simplex-wise $F$-sequence.  \\
- If $\ka_{i} \in \partial(\tau)$, we have  $F(\ka_{i}) \leq F(\tau) = F(\ka_{i+1})$ since $F$ is a stack on $K$. 
\end{enumerate}

\item Suppose $\ka_{i} = (\sigma, \tau)$; that is $K_i$ is an expansion of $K_{i-1}$. 
\begin{enumerate} [topsep=0cm]
\item Suppose $\ka_{i+1}$ is a critical simplex; that is  $K_{i+1}$ is a filling of $K_{i}$. We observe that we cannot have
$\tau \in \partial(\ka_{i+1})$, otherwise there would exist a face $\nu \in K_{i}$, with $\nu \in \partial(\ka_{i+1})$, $\nu \not= \tau$,
such that $\sigma$ is a proper face of~$\nu$. \\
- If $\sigma \not\in \partial(\ka_{i+1})$, then $\diamond \ms'$ is a simplex-wise $F$-sequence. \\
- If $\sigma \in \partial(\ka_{i+1})$, then $F(\ka_{i}) = F(\sigma) \leq F(\ka_{i+1})$ since $F$ is a stack on $K$. 
\item Suppose $\ka_{i+1} = (\sigma', \tau')$; that is $K_{i+1}$ is an expansion of $K_{i}$. \\
- If $\sigma \not\in \partial(\tau')$ and $\tau \not\in \partial(\tau')$, then $\diamond \ms'$ is a simplex-wise $F$-sequence. \\
- If $\sigma \in \partial(\tau')$ or $\tau \in \partial(\tau')$, we have  $F(\ka_{i}) \leq F(\tau') = F(\ka_{i+1})$ since $F$ is a stack on~$K$.
 \hspace*{\fill} $\Box$
\end{enumerate}

\end{enumerate}
\end{proof}

An $F$-sequence is not necessary  a flooding sequence. 
Nevertheless we have the following result which indicates that the gradient vector field of an arbitrary Morse sequence on $F$
may be obtained by computing a flooding sequence on $F$.

\begin{theorem} \label{th:stack4}
Let $F$ be a stack on $K$, and let $\ms$ be an $F$-sequence from $\emptyset$ to $K$. 
Then, there exists a flooding sequence on $F$ which is equivalent to $\ms$.
\end{theorem}

\begin{proof} If  $\ms$ is not a flooding sequence, then there exist two elements $\ka_i$ and $\ka_{i+1}$
in $\diamond \ms$ such that $F(\ka_i) > F(\ka_{i+1})$. By Lemma \ref{prop:stack4}, we can
swap these elements.  Note that this swap corresponds exactly to the swap of the bubble sorting algorithm \cite{cormen01introduction}.
By repeating this operation, we obtain a sorted sequence
 $\overrightarrow{V}$  which is equivalent to $\ms$, and which is a flooding  sequence.
$\Box$
\end{proof}

\begin{remark}
By Lemma \ref{prop:stack4}, a flooding sequence may be obtained from an arbitrary $F$-sequence from $\emptyset$ to $K$ with a sorting algorithm.
In order to effectively obtain a Morse sequence, it can be seen that it is sufficient for the sorting algorithm to be \emph{stable}. 
That is, equal elements must appear in the same order in the sorted sequence as they do in the original $F$-sequence \cite{cormen01introduction}.
\end{remark}

We conclude this section by illustrating
Theorem  \ref{th:stack4} with the sequence $\ms'$ given  Fig. \ref{fig:BasicFlooding}. 
We have $\diamond \ms' = \langle a, (b, ab), c, bc, (d,cd),(e,ed), be,(bd,bde) \rangle$.
By iteratively applying  Lemma \ref{prop:stack4}, we see that we can derive the simplex-wise sequence 
$\diamond \ms''' = \langle c, (d,cd),(e,ed), a, (b, ab), bc, be,(bd,bde) \rangle$:
the corresponding sequence $\ms'''$ is an $F$-sequence which is equivalent to $\ms'$. 
Furthermore, $\ms'''$ is a flooding sequence on $F$.


\section{Flooding sequences and cosimplicial complexes}
\label{sec:floodcosim}

Let $F$ be a stack on $K$ and let  $\ms = \langle \emptyset = K_0,..., K_k = K \rangle$ be an $F$-sequence. 
Let $\diamond \ms = \langle \ka_1, \ldots ,\ka_k \rangle$. 
For each $\lambda \in \bb{Z}$ we write $\diamond \ms [\lambda]$ for the subsequence of $\diamond \ms$ composed 
of all $\ka_i$ such that $F(\ka_i) = \lambda$. 
Recall that a  subsequence is a {\it substring} if it a sequence of
contiguous elements of the original sequence.
If the sequence $\ms$ is a flooding sequence then, for each $\lambda \in \bb{Z}$, the sequence $\diamond \ms [\lambda]$ is a substring of
$\diamond \ms$ which is a simplex-wise Morse sequence from $F_{\lambda-1}$ to $F_\lambda$. 
The set $S = F_{\lambda} \setminus F_{\lambda-1}$ is a cosimplicial complex. From Proposition \ref{prop:cosim1},
we obtain:

\begin{proposition} \label{prop:cosim2} Let $(L,K)$ be a pair of simplicial complexes such that $L \subseteq K$,
and let $S = K \setminus L$. A sequence $\ms$ is a Morse sequence from $L$ to $K$ if and only if
$\ms$ is a Morse sequence from $\underline{S}$ to $\overline{S}$.
\end{proposition}

Let $S$ be a cosimplicial complex. 
If $\overrightarrow{V}$ is a Morse sequence from $\underline{S}$ to~$\overline{S}$, then we say that
$\overrightarrow{V}$ and
$\diamond \overrightarrow{V}$ are \emph{Morse sequences on $S$}.

If $F$ is a stack on $K$, recall that $F [\lambda]  = \{\nu \in K \; \mid \; F(\nu) = \lambda \}$;
 we have $F [\lambda] = F_{\lambda} \setminus F_{\lambda-1}$.
We derive the following result from the previous observations.

\begin{proposition} \label{prop:par1} 
Let $F$ be a stack on $K$ and let  $\ms = \langle \emptyset = K_0,..., K_k = K \rangle$ be an $F$-sequence. Let $[0,H]$ be the range of $F$.
The sequence $\ms$ is a flooding sequence if and only if 
$\diamond \ms$ is a concatenation of the form $\diamond \ms = \diamond \ms [0]$ $... $ $\diamond \ms [\lambda]$ $...$ $\diamond \ms [H]$ 
where, for each $\lambda \in [0,H]$, $\diamond \ms [\lambda]$ is a Morse sequence on $F [\lambda]$. 
\end{proposition}


\begin{definition} \label{def:contraction}
 Let $S$ be a cosimplicial complex and let $\sigma,\tau,\nu \in S$, with $\sigma \in \partial \tau$.
 \begin{itemize}[topsep=0cm] 
\item If $\delta(\sigma, S) = \{ \tau \}$, we say that $S' = S \setminus \{ \sigma, \tau \}$ is a \emph{reduction of $S$},
\item If $\delta(\nu,S) = \emptyset$, we say that $S' = S \setminus \{ \nu \}$ is a \emph{perforation of $S$},
\item If $\partial(\tau,S) = \{ \sigma \}$, we say that $S' = S \setminus \{ \sigma, \tau \}$ is a \emph{coreduction of $S$}, 
\item If $\partial(\nu,S) = \emptyset$, we say that $S' = S \setminus \{ \nu \}$ is a \emph{coperforation of $S$}.
 \end{itemize}
 In each of the above four cases, we say that $S'$ is a \emph{contraction of $S$}.
 \end{definition}

 \noindent
 Reductions and coreductions have been introduced in \cite{Mro09} in the context of $\mathcal{S}$-complexes. In our context, 
 the following property is crucial for an efficient computation of Morse sequences.

\begin{proposition} \label{prop:cosim3} Let $S$ be a cosimplicial complex and let $\sigma,\tau,\nu \in S$. 
\begin{itemize}[topsep=0cm] 
\item $\overline{S} \setminus \{ \sigma, \tau \}$ is a collapse of $\overline{S}$ iff $S \setminus \{ \sigma, \tau \}$ is a reduction of $S$. 
\item $\overline{S} \setminus \{ \nu \}$ is a perforation of $\overline{S}$ iff $S \setminus \{ \nu \}$ is a perforation of $S$.
\item $\underline{S} \cup \{ \sigma, \tau \}$ is an expansion of $\underline{S}$ iff $S \setminus \{ \sigma, \tau \}$ is a coreduction of $S$. 
\item $\underline{S} \cup \{ \nu \}$ is a filling of $\underline{S}$ iff $S \setminus \{ \nu \}$ is a coperforation of $S$. 
 \end{itemize}
\end{proposition}

We note that a contraction of a cosimplicial complex is also a cosimplicial complex. Thus, the following definition makes sense.

Let $S$ be a cosimplicial complex. A \emph{ Morse cosequence on $S$} is a sequence $\overrightarrow{U} = \langle S = S_0,\ldots,S_k = \emptyset \rangle$ 
such that, for each $i \in [1,k]$,  $S_i$ is a contraction of $S_{i-1}$. The sequence is
increasing (resp. decreasing) if each $S_i$ is either a coreduction or a coperforation of $S_{i-1}$
(resp. a reduction or a perforation of $S_{i-1}$). 

Let $L\subseteq K$ be two simplicial complexes. Let
$\ms = \langle L = K_0,\ldots,K_k =K \rangle$ be a sequence of nested complexes. From Propositions \ref{prop:cosim2} and \ref{prop:cosim3}, we obtain: \\
Let $\overrightarrow{U} = \langle S_0,\ldots,S_k \rangle$ with $S_i = K \setminus K_{i}$.
Then $\ms$ is a Morse sequence if and only if  $\overrightarrow{U}$ is an increasing Morse cosequence on $K \setminus L$. \\
Let $\overrightarrow{U} = \langle S_0,\ldots,S_k \rangle$ with $S_i = K_{k-i} \setminus L$.
Then $\ms$ is a Morse sequence if and only if  $\overrightarrow{U}$ is a decreasing Morse cosequence on $K \setminus L$.

\section{Computing a flooding sequence}
\label{sec:compflood}

We have seen that a Morse sequence $\ms$ from $L$ to $K$ may be obtained by computing a Morse sequence $\diamond \overrightarrow{V}$
on the set $S = K \setminus L$. 
The sequence $\diamond  \overrightarrow{V}$ can be constructed from the left to the right through iterative coreductions and coperforations of the set $S$. 
Alternatively, $\diamond  \overrightarrow{V}$ can be constructed from  the right to the left through iterative reductions and perforations of $S$. 
 
Scheme \ref{Sch1} allows us to compute a Morse sequence $\diamond  \overrightarrow{V}$ on an arbitrary cosimplicial complex $S$
by considering the two previous constructions.
At each step of the scheme, the set $S'$ is a contraction of $S$.

\begin{algorithm2e*}[h]
\renewcommand{\algorithmcfname}{Scheme}%
\KwData{A cosimplicial complex $S$ \\}
\KwResult{ {\bf Sequence}$(S)$, which is a simplex-wise Morse sequence on $S$.}
\SetKwFor{WhileBy}{while}{do one of the following}{end while}
\SetKwData{Kw}{Extract $\ka$ such that: \\
$\; \;$ -  either $\ka = (\sig,\tau)$, with $\sig \subset \tau$, and $S' = S \setminus \{ \sigma, \tau \}$ is a reduction or a coreduction of $S$, \\
$\; \;$ - or $\ka = \nu$ and $S' = S \setminus \{ \nu \}$ is a perforation or a coperforation of $S$.}
\SetKwFor{Red}{reduce}{}{end red}
$\overrightarrow{L} := \langle \rangle$; $\overrightarrow{R} := \langle \rangle$\;
\While{$S \not= \emptyset$}{
\Kw \\
\lIf{$S'$ is a coreduction or a coperforation of $S$}{$\overrightarrow{L} := \overrightarrow{L} \cdot \ka$}  
\lIf{$S'$ is a reduction or a perforation of $S$}{$\overrightarrow{R} := \ka \cdot \overrightarrow{R}$}
$S := S'$
}
$\diamond  \overrightarrow{V} := \overrightarrow{L} \cdot \overrightarrow{R}$\;
\Return{$\diamond  \overrightarrow{V}$}

\caption{ 
{\bf Sequence}$(S)$} 
    \label{Sch1}
\end{algorithm2e*}

We consider three cases of Scheme \ref{Sch1}: 
\begin{enumerate} [topsep=0cm]
\item If $\diamond  \overrightarrow{V}$ is constructed using only coreductions and coperforations, then the scheme is an
\emph{increasing scheme}. We say that $\diamond  \overrightarrow{V}$ is a \emph{maximal sequence} if we make a coperforation only when no coreduction can be made. 
\item If $\diamond  \overrightarrow{V}$ is constructed using only reductions and perforations, then the scheme is a 
\emph{decreasing scheme}. We say that $\diamond  \overrightarrow{V}$ is a \emph{minimal sequence} if we make a perforation only when no reduction can be made. 
\item Otherwise we say that the scheme is an \emph{interleaved scheme}. We say that $\diamond  \overrightarrow{V}$ is a \emph{min/max sequence} if we make a perforation or a coperforation
only when no reduction and no coreduction can be made. \\
\end{enumerate}

Let us illustrate these three schemes with the complex $K$ of Figure \ref{fig:BasicFlooding} (a). 
\begin{enumerate} [topsep=0cm]
\item The Morse sequence $\ms$ corresponding to Figure \ref{fig:BasicFlooding} (b)
is such that: \\
\hspace*{\fill}  $\diamond \ms = \langle a, (b, ab), (c,bc), (d,cd),(e,ed),be,(bd,bde) \rangle$. \hspace*{\fill} \\
This sequence is a maximal sequence. It can be built from the left to the right with Scheme \ref{Sch1}.
Under this ordering, no coreduction can be made when the two coperforations $a$ and $be$ are done. 
\item Let $\ms'''$ be the Morse sequence such that: \\
\hspace*{\fill}  $\diamond \ms''' = \langle b, (c,bc), (d,cd),(e,ed),be,(a,ab), (bd,bde) \rangle$. \hspace*{\fill} \\
This sequence is a minimal sequence. It can be built from the right to the left with Scheme \ref{Sch1}.
Under this ordering, no reduction can be made when the two perforations $be$ and $b$ are done. 
Note that $\ms$ is not a minimal sequence and $\ms'''$ is not a maximal sequence.
\item For a min/max sequence, at the second step of Scheme \ref{Sch1}, we must have either $\overrightarrow{L} = \langle \rangle$ and
 $\overrightarrow{R} = \langle (a,ab),(bd,bde) \rangle$ or $\overrightarrow{L} = \langle \rangle$ and
 $\overrightarrow{R} = \langle (bd,bde),(a,ab) \rangle$.
 At the third step, we can choose  to make the coperforation $b$. 
 In this case, we must make three coreductions at the following steps.
 Thus, at line 9 of Scheme \ref{Sch1}, we may obtain: \\
 \hspace*{\fill} $\overrightarrow{L} = \langle b, (c,bc), (d,cd),(e,ed),be \rangle$ and
 $\overrightarrow{R} = \langle (a,ab), (bd,bde) \rangle$ \hspace*{\fill}. \\
We retrieve the sequence $\ms'''$ after the concatenation of $\overrightarrow{L}$ and $\overrightarrow{R}$.\\
\end{enumerate}

The purpose of maximal, minimal, and min/max sequences
is to try to minimize the number of critical
simplexes.  This problem is, in general, NP-hard \cite{Jos06}.
Therefore, these sequences do not, in general, yield optimal results. 
Note also that different numbers of critical simplexes may be obtained by
these three types of sequences,
see Figure 11 of \cite{fugacci2019computing} and Figure 4 of \cite{Bertrand2023MorseSequences} for examples which illustrate this fact.

 The computational models of increasing and decreasing schemes are classical models for extracting gradient vector fields, see
 \cite{Ben14}, \cite{Mro09}, \cite{Har14}, \cite{fugacci2019computing}. Also the interleaved model is introduced in  \cite{fugacci2019computing}.
 A distinctive feature of our approach is the use of these schemes to compute Morse sequences. 
A Morse sequence on a complex $K$ not only defines a gradient vector field on $K$, but also provides a specific structure on $K$.

With Proposition \ref{prop:par1}, we derive Scheme \ref{Sch2} which is a direct extension of Scheme \ref{Sch1} for computing flooding sequences. 
Given a stack $F$ on $K$, we consider the different sections $F [\lambda] $ of $F$ at level $\lambda$.
 For each  $\lambda$,
a Morse sequence $\diamond  \overrightarrow{W} [\lambda]$ on $F [\lambda] $ is computed thanks to 
Scheme \ref{Sch1}. Then, the sequences $\diamond  \overrightarrow{W} [\lambda]$ are simply concatenated in order 
to obtain a flooding sequence on $F$.

\SetKwFor{ForBy}{for}{do in parallel}{end for}

\begin{algorithm2e*}[h]
\renewcommand{\algorithmcfname}{Scheme}%
\KwData{A simplicial complex $K$ and a stack $F$ on $K$; the range of $F$ is $[0,H]$. \\
We set
$F[\lambda] := \{\nu \in K \; | \; F(\nu) = \lambda \}$.
\\}
\KwResult{{\bf Flood$(F)$}, which is a simplex-wise flooding sequence on $F$.}

\ForBy{$\lambda \in[0,H]$}{
 $\diamond  \overrightarrow{W} [\lambda]:= \; ${\bf Sequence}$(F [\lambda])$}

$\diamond \overrightarrow{W}  := \langle \rangle$; $\lambda :=0$\;
\While{$\lambda \leq H$}
{
$\diamond \overrightarrow{W}  := \diamond \overrightarrow{W} \; \cdot \; \diamond \overrightarrow{W} [\lambda]$; \\
$\lambda := \lambda+1$
}

\Return{$\diamond \overrightarrow{W}$}

\caption{{\bf Flood}($F$) }\label{Sch2}
\end{algorithm2e*}

With the three above instances of Scheme \ref{Sch1}, we obtain three instances of Scheme \ref{Sch2}
which allow to compute maximal, minimal, and min/max flooding sequences.

Let us consider again the stack $F$ given  Figure \ref{fig:BasicFlooding} (c).
We have: 

\hspace*{\fill}  $F[1] = \{ c\}$, $F[3] = \{d,e,de, cd\}$, $F[5] = \{a,b,ab,bc,be\}$, $F[8] = \{bd,bde\}$. \hspace*{\fill} \\
With Scheme~\ref{Sch1}, we may obtain the following maximal sequences: 

- {\bf Sequence}$(F [1]) = \langle c \rangle$,

- {\bf Sequence}$(F [3]) = \langle (d,cd),(e,ed) \rangle$,

- {\bf Sequence}$(F [5]) = \langle (b,bc), (a, ab), be \rangle$,

- {\bf Sequence}$(F [8]) = \langle (bd,bde) \rangle$.\\
Thus, with Scheme \ref{Sch2}, we may obtain precisely the flooding sequence $\diamond \ms'' $
which is illustrated Figure \ref{fig:BasicFlooding} (e).
Observe that, in this simple case, we obtain an {\em optimal $F$-sequence}, that is, an $F$-sequence with the smallest number of critical faces. 
Observe also that, by Theorem \ref{th:stack4}, an optimal $F$-sequence may always be obtained by considering solely flooding sequences.

With Scheme \ref{Sch1}, it can be seen that a 
maximal, minimal, or  min/max sequence $\diamond  \overrightarrow{W} [\lambda]$ of Scheme \ref{Sch2} may be obtained with a time complexity $\mathcal{O}(d N_\lambda)$,
where $d$ is the dimension of the complex $K$, and $N_\lambda$  is the number of simplexes of $F [\lambda] $, see~\cite{BN25}.
Therefore, since the sets $F [\lambda]$ are disjoint, all the sequences $\diamond  \overrightarrow{W} [\lambda]$ can be computed with a time complexity $\mathcal{O}(d N)$,
where $N$  is the number of simplexes of $K$.
Now, if each sequence is stored as a list, and if the range of $F$ is $[0,H]$, it can be seen that the concatenations may be done with complexity $\mathcal{O}(H)$.  Since $H$ is bounded by $N$, we obtain $\mathcal{O}(d N)$ for computing a flooding sequence with a sequential algorithm.

Also the computations of all sequences may be done in parallel. 
In the best case, that is, if the values of the simplexes are equally distributed, we achieve $\mathcal{O}(\frac{d N}{H} + H)$
 for computing the flooding sequence.

\section{Conclusion}

This work extends the framework of Morse sequences to stacks on simplicial complexes, introducing a flexible and expressive model for discrete Morse theory. By incorporating stack-based constraints, we define flooding sequences—a refined version of Morse sequences that respect monotonic ordering, reflecting processes commonly observed in watershed algorithms.

We have shown that the gradient vector field arising from any Morse sequence on a stack can be recovered by computing an equivalent flooding sequence, underscoring the generality of the flooding model. 
We also proposed general construction schemes, based on cosimplicial complexes, that provide efficient methods for computing flooding sequences.
It should be noted that all the notions introduced in the paper may be directly adapted to more general complexes, for example to cubical complexes.

Importantly, a Morse sequence encodes richer structural information about a complex than a gradient field alone. In particular,  
the Morse complex of an object can be easily derived via the reference map introduced in~\cite{bertrand2025morse}, requiring only a linear scan of the sequence.

It is also worth emphasizing that a stack is not necessarily an injective function—different faces in the complex may share the same value. This contrasts with some approaches that require distinct values for each face \cite{fugacci2019computing,robins2011theory}. Allowing arbitrary stacks makes the model more suitable for real-world data, where such constraints are typically not satisfied \cite{ROCCA2025101299}.

Overall, the Morse sequence framework opens promising directions for future research, including extensions to persistent homology as well as applications in image segmentation and topological data analysis.\\

{\bf Acknowledgements}
The author wishes to express his thanks to Laurent Najman
for his active interest in this paper and for stimulating conversations.


\bibliographystyle{splncs04}
\bibliography{biblio}

\end{document}